\documentclass{eptcs}
\usepackage{amssymb}
\usepackage{amsthm}
\usepackage{amsmath}
\usepackage{algorithm}
\usepackage{tikz}
\usepackage{algpseudocode}

\title{Computing the Complete Pareto Front}
\author{R\"udiger Ehlers
	\institute{University of Bremen \& DFKI GmbH}
}

\newcommand{\newterm}{\emph}
\newcommand{\BB}{\mathbb{B}}
\newcommand{\TRUE}{\mathbf{true}}
\newcommand{\FALSE}{\mathbf{false}}
\newtheorem{example}{Example}
\newtheorem{lemma}{Lemma}
\newtheorem{corollary}{Corollary}

\begin{document}
\maketitle

\begin{abstract}
We give an efficient algorithm to enumerate all elements of a Pareto front in a multi-objective optimization problem in which the space of values is finite for all objectives. 

Our algorithm uses a feasibility check for a search space element as an oracle and minimizes the number of oracle calls that are necessary to identify the Pareto front of the problem. Given a $k$-dimensional search space in which each dimension has $n$ elements, it needs $p \cdot (k \cdot \lceil \log_2 n \rceil + 1) + \psi(p)$ oracle calls, where $p$ is the size of the Pareto front and $\psi(p)$ is the number of greatest elements of the part of the search space that is not dominated by the Pareto front elements. We show that this number of oracle calls is essentially optimal as approximately $p \cdot k \cdot \log_2 n$ oracle calls are needed to identify the Pareto front elements in sparse Pareto sets and $\psi(p)$ calls are needed to show that no element is missing in the set of Pareto front elements found.
\end{abstract}

\section{Introduction}
In multi-objective optimization, we explore the \newterm{Pareto front} of an optimization problem. It consists of all solutions that cannot be improved by some objective without sacrificing their quality in some other objective. 

Due to the fact that many optimization problems in practice have multiple objectives, computing \newterm{Pareto-optimal} solutions (i.e., the elements of the Pareto front) is a classical topic in constraint programming, mathematical optimization, and operations research \cite{NonLinearOpBook,DBLP:conf/aaai/HartertS14,DBLP:conf/dagstuhl/2008moo}. Most algorithms to explore the Pareto front focus on \emph{sampling} it, i.e., computing some  representative elements of the Pareto front \cite{Shan2004}. 
Others use the properties of the concrete optimization problem to derive the complete Pareto front \cite{DBLP:journals/eor/OkamotoU11}. For example, in linear programming with multiple optimization objectives, the set of Pareto-optimal solutions forms a set of polyhedra that can be represented in an explicit way.

If the space of possible solutions is finite, \emph{enumerating} the set of Pareto-optimal solutions can be feasible, even if the structure of the optimization problem is unknown and only an \emph{oracle} is available that returns for some search space element whether it is feasible, i.e., whether it is a possible solution to the optimization problem. Elbassioni \cite{ElbassioniThesis} gives an algorithm to compute all Pareto-optimal solutions using an oracle for a \emph{submodular} optimization problem. The submodularity requirement essentially amounts to the set of solutions being convex or concave. The algorithm requires a number of oracle calls that is quasi-polynomial in $k$, $n$, and $p$, where $k$ is the number of objectives, $n$ is the the number of possible values for each objective, and $p$ is the number of Pareto-optimal solutions. Elbassioni shows that many optimization problems over graphs and multi-objective linear programming problems are submodular, and hence his algorithm can be applied.
A more general algorithm that does not require the submodularity of an optimization problem was not yet known, however.

\subsection{Contribution}

In this paper, we give an efficient algorithm to compute the Pareto front of a multi-objective optimization problem. It requires the set of solutions to be finite and uses an oracle to evaluate whether a solution is feasible or not. The problem does not need to be submodular and the only thing that needs to be known about it is that it is \newterm{monotone}, i.e., whenever a solution is feasible, then all other solutions that are \newterm{dominated} by it also need to be feasible.

Our algorithm minimizes the number of calls to the oracle and requires $p \cdot (k \cdot \lceil \log_2 n \rceil + 1) + \psi(p)$ many calls, where $p$ is the number of Pareto-optimal solutions, $k$ is the number of objectives, and $n$ is the number of elements in each dimension. The term $\psi(p)$ represents the size of the \newterm{co-Pareto set}, i.e., the number of greatest elements of the part of the search space that contains the infeasible solutions.
We show that our algorithm is essentially optimal by proving that
\begin{enumerate}
\item when dividing the search space into $(\frac{n}{\log_2 n})^k$ large regions such that in each region, there is at most one Pareto front element, no two regions with Pareto front elements have comparable elements, and for all of the $p$ Pareto front elements, the region that it contains is known, $p \cdot k \cdot \log_2 \left(\frac{n}{\log_2 n} \right)$ calls to the oracle are necessary in order to compute the exact positions of the Pareto front elements, and
\item when $p$ Pareto front elements have been found, $\psi(p)$ calls to the oracle need to be made in order to check that no Pareto front element went unnoticed.
\end{enumerate}
Our algorithm is particularly well-suited to explore Pareto fronts if the number of Pareto-optimal solutions is substantially smaller than $n$. It is also \newterm{any-time}, meaning that it computes the Pareto points successively, so that they can be reported to the user one-by-one, without the need to wait for the algorithm's completion.

Implementations of the algorithm presented in this paper in the programming languages \texttt{Python} and \texttt{C++} can be obtained from \url{https://github.com/progirep/ParetoFrontEnumerationAlgorithm}. The implementations are available under the permissive MIT open source licence.

\subsection{Structure of the Paper}

In the next section, we give the preliminaries. Our Pareto front enumeration algorithm is given in Section~\ref{sec:algorithm}. The section afterwards contains a discussion of the algorithm, including a correctness proof and a complexity analysis. We conclude in Section~\ref{sec:conclusion}.

\section{Preliminaries}
Let $[n]^k = \{0,1, \ldots, n-1,n\}^k$ be the \newterm{search space}, and let $\leq_k$ be a partial order over elements from $[n]^k$ that compares two elements $\vec x = (x_1, \ldots, x_k)$ and $\vec x' = (x'_1, \ldots, x'_k)$ element-wise, i.e., for which we have $\vec x \leq_k \vec x'$ if and only if for all $i \in \{1,\ldots,k\}$, we have $x_i \leq x'_i$. If $\vec x, \vec x' \in [n]^k$ with $\vec x \neq \vec x'$ and $\vec x \leq_k \vec x'$, we say that $\vec x$ is \newterm{smaller} than $\vec x'$, or conversely, that $\vec x'$ is \newterm{greater} than $\vec x$. If $\vec x$ is neither equal to, greater than, or smaller than $\vec x'$, then we say that $\vec x$ and $\vec x'$ are \newterm{incomparable}. If $\vec x$ is smaller than $\vec x'$, then we also say that $\vec x$ \newterm{dominates} $\vec x'$. We say that a set $S \subseteq [n]^k$ is an \newterm{anti-chain} if the elements from $S$ are pair-wise incomparable.

A function $f : [n]^k \rightarrow \BB$ is called a \newterm{feasibility function} if it is \newterm{monotone}, i.e., if for all $\vec x$ with $f(\vec x) = \TRUE$, we have that $f(\vec x') = \TRUE$ holds for all $\vec x'$ that are greater than $\vec x$. We say that a \newterm{point} $\vec x \in [n]^k$ is \newterm{Pareto-optimal} if $f(\vec x) = \TRUE$ and for every $\vec x'$ that is smaller than $\vec x$, we have $f(\vec x') = \FALSE$. 
The set of all Pareto optima of a function $f$ is also called the \newterm{Pareto front}. It forms an anti-chain by definition.

We will also later need the notion of a \newterm{co-Pareto point}. An element $\vec x$ of $[n]^k$ is called a \newterm{co-Pareto point} if $f(\vec x) = \FALSE$ and for all $\vec x'$ that are greater than $\vec x$, we have $f(\vec x') = \TRUE$. The set of all co-Pareto points is also called the \newterm{co-Pareto front}.
We say that some point $\vec x$ lies \newterm{between} some point $\vec y$ and $(0)^k$ if we have $(0)^k \leq_k \vec x \leq_k \vec y$.

\paragraph{Note:}
For the simplicity of notation, we assumed that all optimization objectives have the same domain. The enumeration algorithm to the presented in the next section is easy to adapt to optimization objectives with different domains, though.

\section{Exploring the Pareto Front}
\label{sec:algorithm}
The design objective of the algorithm to enumerate all elements of a given feasibility function's Pareto front is to keep the number of evaluations of the function as low as possible. It enumerates the Pareto front in a point by point fashion and maintains a set $S$ of maximal elements of the part of the search space that still has to be checked for if it contains more Pareto points. Whenever there is another Pareto point, then it is assured that it dominates or is equal to one of the points in $S$. % as that set is maintained in a way such that no point in the search space that  dominates or is equal to an element in $S$ is dominated by a Pareto point that has already been found. At the same time, 
The algorithm makes sure that $S$ is kept complete enough not to ever miss a Pareto point.

\begin{algorithm}[t]
  \newcommand{\LineIf}[2]{     \State \algorithmicif\ {#1}\ \algorithmicthen\ {#2} }
   \begin{algorithmic}[1]
 	\Function{ParetoEnumerate}{$n$,$k$,$f$}
 	\State $S \gets \{ (n)^k \}$
 	\State $P \gets \emptyset$
 	\While{$S \neq \emptyset$}
 		\State pick (without removing) some $\vec x$ from $S$
 		\If{$f(\vec x) = \TRUE$} \label{line:checkSearchPoint}
 			\State $\vec x \gets $\Call{SearchParetoPoint}{$\vec x,k,f$} \label{line:searchParetoPoint}
 			\State $P \gets P \cup \{\vec x\}$ \label{line:updateSAndPStart}
 			\State $S' \gets \emptyset$
 			\For {$\vec y \in S$} \label{line:SUpdateLoopStart}
	 			\If{$\neg (\vec x \leq_k \vec y)$}
 					\State $S' \gets S' \cup \{ \vec y\}$
 				\Else
 					\For{$i \in \{1, \ldots, k\}$} \label{line:updateLoopStart}
 						\If{$x_i > 0$}
	 						\State $S' \gets S' \cup \{ (y_1, \ldots, y_{i-1},x_i - 1, y_{i+1}, \ldots, y_k) \}$
	 					\EndIf 
 					\EndFor \label{line:updateLoopEnd}
 				\EndIf
 			\EndFor \label{line:SUpdateLoopEnd}
 			\State $S \gets \Call{RemoveDominatingElements}{S',k}$ \label{line:updateSAndPEnd}
 		\Else
 			\State $S \gets S \setminus \{\vec x\}$
 		\EndIf
 	\EndWhile
	\State \Return $P$ 	
	\EndFunction
   \end{algorithmic}
   \caption{\label{algo:paretoEnumerator} Pareto point enumeration algorithm}
\end{algorithm}

Algorithm~\ref{algo:paretoEnumerator} shows the main function of the Pareto front exploration approach. The set of Pareto points that have already been found is stored in the variable $P$. The algorithm uses the two sub-functions \textsc{SearchParetoPoint} and \textsc{RemoveDominatingElements}, which are given separately in Algorithm~\ref{algo:searchParetoPoint} and Algorithm~\ref{algo:RemoveDominatingElements}.

\begin{algorithm}[t]
  \newcommand{\LineIf}[2]{     \State \algorithmicif\ {#1}\ \algorithmicthen\ {#2} }
   \begin{algorithmic}[1]
 	\Function{SearchParetoPoint}{$\vec x$,$k$,$f$}
	\For{$i \in \{1, \ldots, k\}$}
		\State $\mathit{max} \gets x_i + 1$
		\State $\mathit{min} \gets 0$
		\While{$\mathit{max}-\mathit{min} > 1$} \label{line:binarySearchStart}
			\State $\mathit{mid} \gets \mathit{min} + \lfloor \frac{\mathit{max} - \mathit{min} - 1}{2}\rfloor$ \label{line:testPoint}
			\State $x_i \gets \mathit{mid}$
			\If{$f(\vec x) = \TRUE$}
				\State $\mathit{max} \gets \mathit{mid}+1$
			\Else
				\State $\mathit{min} \gets \mathit{mid}+1$
			\EndIf
		\EndWhile \label{line:binarySearchEnd}
		\State $x_i \gets \mathit{min}$
	\EndFor
	\State \Return $\vec x$
	\EndFunction
   \end{algorithmic}
   \caption{\label{algo:searchParetoPoint} A sub-function to find a Pareto point that dominates a feasible solution $\vec x$ (or is $\vec x$ if no feasible solution dominates $\vec x$) based on binary search.}
\end{algorithm}

\begin{algorithm}[t]
  \newcommand{\LineIf}[2]{     \State \algorithmicif\ {#1}\ \algorithmicthen\ {#2} }
   \begin{algorithmic}[1]
 	\Function{RemoveDominatingElements}{$S$,$k$}
	\State $S' \gets \emptyset$
	\For{$\vec x \in S$}
		\State $\mathit{found} \gets \FALSE$
		\For{$\vec y \in S$}
			\If{$(\vec x \leq_k \vec y) \wedge (\vec x \neq \vec y)$}
				\State $\mathit{found} \gets \TRUE$
			\EndIf 
		\EndFor	
		\If{$\neg \mathit{found}$}
			\State $S' \gets S' \cup \{ \vec x \}$
		\EndIf
	\EndFor
	\State \Return $S'$
	\EndFunction
   \end{algorithmic}
   \caption{\label{algo:RemoveDominatingElements} A sub-function to remove dominating elements from $S$}
\end{algorithm}

The function \textsc{ParetoEnumerate} initializes $S$ to contain exactly the greatest element of the search space -- all Pareto points are definitely smaller than or equal to this point. 
In the main loop of the algorithm, it takes an element $\vec x$ from $S$ and checks if the space between $(0)^k$ and $\vec x$ contains any feasible solution for $f$. Due to the monotonicity of $f$, this is the case if and only if $f(\vec x)=\TRUE$. If $f(\vec x)=\FALSE$, then $\vec x$ is removed from $S$, as the part of the search space between $(0)^k$ and $\vec x$ then no longer has to be checked for containing a Pareto point.

The algorithm maintains the invariant that no point in $S$ is dominated by a point in $P$, and no point in $P$ is dominated by a point in $S$. Thus, if $f(\vec x)$ is found to be $\TRUE$, then we already know that some Pareto point is missing from $P$, and this point is either $\vec x$, or we have that $\vec x$ is dominated by some Pareto point that is not yet in $P$.

To find the Pareto point, Algorithm~\ref{algo:paretoEnumerator} calls the sub-function \textsc{SearchParetoPoint}. It takes the feasible solution $\vec x$ and finds \emph{one} Pareto point $\vec y$ such that $\vec y \leq_k \vec x$. It performs a binary search in each dimension. So in each iteration of its main loop, the $i$th element of $\vec x$ is minimized while retaining the fact that $f(\vec x) = \TRUE$. As we thus know that at the end of the main loop, no element of $x_1, \ldots, x_k$ can be reduced by one without losing the property that $\vec x = \TRUE$, we know that the function returns a Pareto point.

In lines~\ref{line:updateSAndPStart}--\ref{line:updateSAndPEnd}, Algorithm~\ref{algo:paretoEnumerator} updates the sets $P$ and $S$. The  newly found Pareto point is added to $P$, and $S$ is altered such that (1) the invariant holds that no point in $S$ is dominated by one in $P$, and that (2) $S$ contains the maximal elements of the search space fraction for which it is not yet known if it contains Pareto points. For altering $S$, the loop from lines~\ref{line:SUpdateLoopStart}--\ref{line:SUpdateLoopEnd} iterates over the old points in $S$ and computes new points that it stores in $S'$. The set $S'$ is then stripped from all non-maximal elements in line \ref{line:updateSAndPEnd} and the result is copied back to $S$. This ensures that $S$ stays an anti-chain.

Set $S'$ is computed successively from the points $\vec y$ in $S$. If we do not have $\vec x \leq_k \vec y$, where $\vec x$ is the newly found Pareto point, then $\vec y$ may still be a Pareto point. Also, $\vec y$ does not dominate $\vec x$ in this case. Thus, $\vec y$ can remain untouched and is copied straight to $S'$. If on the other hand we have $\vec x \leq_k \vec y$, then $\vec y$ cannot remain in $S$, as it is dominated by $\vec x$, for whom we know that $f(\vec x) = \TRUE$. What we do not know at that point, however, is whether the search space has more Pareto points between $(0)^k$ and $\vec y$ apart from $\vec x$. For such other Pareto points, we know that they must have an objective value that is smaller than the respective objective value in $\vec x$ in at least one dimension, as otherwise $\vec x$ would be dominated by another feasible solution and hence would not be a Pareto point. Lines \ref{line:updateLoopStart} to \ref{line:updateLoopEnd} now compute the largest elements of the space between $\vec y$ and $(0)^k$ in which such an additional Pareto point could lie, taking into account the fact that the new point must have such a lower value in at least one objective.

After the new set of maximal elements of the part of the search space that may contain more Pareto points have been computed, elements that dominate other elements in this set are removed as they are redundant. This task is performed by subfunction \textsc{RemoveDominatingElements}. The result is copied back to $S$. The algorithm terminates when the set $S$ becomes empty, which means that there are no more parts of the search space to be explored.

\section{Discussion of Algorithm~\ref{algo:paretoEnumerator}}

Before discussing the main properties of Algorithm~\ref{algo:paretoEnumerator}, let us consider an example execution of it.

\begin{example}
Let $n = 3$, $k=3$, and $f : [n]^k \rightarrow \BB$ be defined as:
\begin{equation*}
f(\vec x) = \begin{cases}
\TRUE & \text{if } \vec x \geq_k (2,1,1) \\
\TRUE & \text{if } \vec x \geq_k (1,2,2) \\
\FALSE & \text{else}
\end{cases}
\end{equation*}
From the definition of $f$ in this example, we can immediately see that the feasibility function $f$ has two Pareto points in $[n]^k$. The algorithm initializes $S$ to $\{(3,3,3)\}$, which is dominated by all Pareto points. In line \ref{line:checkSearchPoint}, the algorithm then determines that a Pareto point is missing. The call to \textsc{SearchParetoPoint} then returns the lexicographical minimal Pareto point that is still missing. This is $(1,2,2)$ here, so $P$ is set to $\{(1,2,2)\}$. In lines \ref{line:SUpdateLoopStart} to \ref{line:SUpdateLoopEnd}, the set $S$ is then updated. As $(1,2,2) \leq_k (3,3,3)$, where the latter is the only element in $S$, lines \ref{line:updateLoopStart} to \ref{line:updateLoopEnd} are executed. As a result, the elements $\{(0,3,3), (3,1,3), (3,3,1)\}$ are stored into $S'$. As $S'$ is already an anti-chain, it is copied to $S$ in line \ref{line:updateSAndPEnd} without modification.

Let us assume that in the next iteration of the main loop of Algorithm~\ref{algo:paretoEnumerator}, the point $\vec x = (3,1,3)$ is picked from $S$. The algorithm finds that $f(\vec x) = \TRUE$. Since the Pareto point found earlier does not dominate any of the points in the search space between $(0,0,0)$ and $(3,1,3)$, the function \textsc{SearchParetoPoint} will find another Pareto point, namely $(2,1,1)$, which is the new value for $\vec x$. Afterwards, Algorithm~\ref{algo:paretoEnumerator} modifies the set $S$. Element $(0,3,3)$ gets copied to $S'$ straight as it is incomparable to $\vec x$. Both other elements  $(3,1,3)$ and $(3,3,1)$ are however dominated by $\vec x$, so they are both processed by lines~\ref{line:updateLoopStart} to \ref{line:updateLoopEnd} of the algorithm. For $(3,1,3)$, the algorithm adds the points $(1,1,3)$, $(3,0,3)$, and $(3,1,0)$ to $S'$, while $(3,3,1)$ is processed to $(1,3,1)$, $(3,0,1)$, and $(3,3,0)$. From the in total 7 points that are now in $S'$, function \textsc{RemoveDominatingElements} removes $(3,0,1)$ and $(3,1,0)$, as they are dominated by the points $(3,0,3)$ and $(3,3,0)$, respectively, that are also in $S'$. Thus we have $S = \{ (0,3,3), (1,1,3),(1,3,1),(3,0,3),(3,3,0) \}$ now. 

The algorithm then proceeds with the main loop and continues to check if $f(\vec x)$ holds for any $\vec x \in S$. This is not the case, so the algorithm terminates while returning exactly the two Pareto points used in the definition of $f$.
\end{example}

The example shows that when a new Pareto point is found, then it may dominate more than one point in $S$. Hence, the number of elements in $S$ can grow by more than $k-1$ in a single iteration of the main loop. Nevertheless, $S$ always contains the maximal elements of the part of the search space that can still contain Pareto points, which we will prove in Lemma~\ref{lem:SandPInteraction}. For each point $\vec x$ in $S$, unless for some $\vec y$ with $\vec y \leq_k \vec x$, $f(\vec y)$ is found to be $\TRUE$ at some point, we have to check if $f(\vec x)=\TRUE$, as there may be another Pareto point sitting at \emph{exactly} $\vec x$. Testing $f$ precisely at this point $\vec x$ is the most efficient way of probing if any point in the search space between $(0)^k$ and $\vec x$ contains a Pareto point.

\subsection{Proof of Correctness}

Let us now formally prove that Algorithm~\ref{algo:paretoEnumerator} computes all Pareto optima of a given feasibility function $f$. We will prove this result in a bottom-up fashion.

\begin{lemma}
Given some $\vec x$ with $f(\vec x) = \TRUE$, \textsc{SearchParetoPoint} computes a Pareto point $\leq_k \vec x$.
\label{lem:searchParetoPointWorks}
\end{lemma}
\begin{proof}
Note that lines \ref{line:binarySearchStart} to \ref{line:binarySearchEnd} of function \textsc{SearchParetoPoint} perform a binary search for $\mathit{mid}$ over values $\leq \mathit{min}$ but $< \mathit{max}$. For the values of $\mathit{mid}$, it tests if $f(x_1, \ldots, x_{i-1}, \mathit{mid}, x_{i+1}, \ldots, x_k)=\TRUE$ and finds the least $\mathit{mid}$ for which this holds. The search process only works correctly if it is known that $f(x_1, \ldots, x_{i-1}, \mathit{max}-1, x_{i+1}, \ldots, x_k)=\TRUE$, and due to the way in which the test points are calculated in line \ref{line:testPoint}, it never calls $f$ on the point $(x_1, \ldots, x_{i-1}, \mathit{max}-1, x_{i+1}, \ldots, x_k)$.

The outer loop of function \textsc{SearchParetoPoint} maintains the invariant that $f(\vec x) = \TRUE$ and for no $j < i$, we have $f(x_1, \ldots, x_{j-1},x_j - 1, x_{j+1}, \ldots, x_k) = \TRUE$. As \textsc{SearchParetoPoint} is only called with $\vec x$ for which $f(\vec x)$ is $\TRUE$, the invariant is initially true. In every loop, the application of binary search then sets element $x_i$ to the least value that maintains the invariant, using the monotonicity of $f$. At the end of the outer loop, all dimensions have been worked through, which ensures that $\vec x$ is a Pareto point -- there is no dimension in which $\vec x$ can be decreased without losing the fact that it is a feasible solution.
\end{proof}

\begin{lemma}
\label{lem:SandPInteraction}
After (and before) every iteration of the main loop of Algorithm~\ref{algo:paretoEnumerator}, $S$ and $P$ are disjoint, $S \cup P$ always forms an anti-chain, and there does not exist a Pareto point $\vec x$ that is not in $P$ and for which for no $\vec y \in S$, we have $\vec x \leq_k \vec y$. At the same time, $P$ contains only Pareto optima.
\end{lemma}
\begin{proof}
We prove the claim by induction over the iteration of the algorithm's main loop and assume that function \textsc{RemoveDominatingElements} does what its name suggests. Due to the simplicity of that function, we omit the proof that it is correct.

Before the main loop, the claim is trivially true: $S$ contains only a single element, and $P$ does not contain an element. Thus, $S \cup P$ is an anti-chain and $S$ and $P$ are disjoint. Since the only element in $S$ is the maximal one of the search space, there cannot be a Pareto point that does not dominate it (or is the maximal element itself).

During the run of the algorithm, the algorithm takes an element $\vec x$ of $S$ and tests if $f(\vec x) = \TRUE$. If this is not the case, it removes $\vec x$ from $S$. This does not change the fact that $S \cup P$ is an antichain. Also, removing $S$ does not change the facts that $S$ and $P$ stay disjunct and that all Pareto points not yet found dominate (or are) points in $S$, as if $f(\vec x) = \FALSE$, we know that there is no feasible solution between $(0)^k$ and $\vec x$, so there is in particular no Pareto point in between.

So let us now assume that in line~\ref{line:checkSearchPoint} of the algorithm, $f(\vec x)$ is found to be $\TRUE$.
As $S$ and $P$ form an anti-chain and $f$ is monotone, this means that there is some Pareto point between $(0)^k$ and $\vec x$ that is not yet in $P$. Function~\textsc{SearchParetoPoint} then computes a new Pareto point $\vec x'$ from it (Lemma~\ref{lem:searchParetoPointWorks}). Adding $\vec x'$ to $P$ retains the property that $P$ only contains Pareto points. However, $S \cup P$ is now not an anti-chain any more. 

This is fixed in lines \ref{line:SUpdateLoopStart} to \ref{line:updateSAndPEnd}. In the last of these lines, $S$ is updated by calling \textsc{RemoveDominatingElements} on some set of points $S'$. This guarantees that $S$ is an anti-chain. Set $P$ is also already an anti-chain as it only contains Pareto points. If all elements of the new content of $S$ are incomparable to all elements of $P$, then this guarantees that $S \cup P$ is also an antichain and that $S$ and $P$ are disjoint. This incomparability is ensured if the new content of $S$ are the maximal elements of the search space part that still has to be tested for if it contains some Pareto optimum, as no point in $S$ can be dominated by one in $P$ in such a case, and no point in $P$ can be dominated by one in $S$, as this would mean that a point in $P$ is not a Pareto point. For the inductive step, we need to prove this anyway, so we will do so next.

Since \textsc{RemoveDominatingElements} removes all dominating elements from $S'$ before storing the result to $S$, it suffices to prove that (a) the maximal elements of the part of the search space that may still contain Pareto optima (apart from the ones already found to be in $P$) are contained in $S'$, and (b) that no point in $S'$ is dominated by one in $P$. By the inductive hypothesis, we can assume assume this to be the case for $S$ and $P$ without $\vec x$. 

Let us first consider part (a). We can assume that $S$ contains the maximal elements of the part of the search space that still has to be considered before $\vec x$ has been found. Every point $\vec y$ in $S$ that is incomparable to $\vec x$ may still contain a Pareto optimum, so it needs to be in $S'$. As the algorithm copies such points straight to $S'$, this is ensured. If a point $\vec y$ in $S$ is dominated by $\vec x$, then all further Pareto points that are $\leq_k \vec y$ need to be incomparable to $\vec x$, which means that they need to have an objective value that is smaller than $\vec x$ in at least one dimension. Algorithm~\ref{line:checkSearchPoint} computes the maximal elements of the search space part between $(0)^k$ and $\vec y$ by iterating over the dimensions and computing the respective points. All of them are added to $S'$. So for each point $\vec y$ in $S$, the search space part between $(0)^k$ and $\vec y$ is in a sense \emph{cleaned} by the ones dominated by $\vec x$. By the inductive hypothesis, $S$ contains \emph{all} of the maximal points of the part of the search space that still has to be considered, and since we cleaned the search space in a conservative way, the claim from part (a) holds.

For part (b), by the inductive hypothesis, all elements in $S$ are incomparable to $P$ without $\vec x$. Elements are copied directly to $S'$ only if they are incomparable to $\vec x$, so they cannot be dominated by any point in $P$. Now consider that one of the points $(y_1, \ldots, y_{i-1},x_i-1,y_{i+1}, \ldots, y_k)$ added to $S'$ is dominated by some element in $P$. This element cannot be $\vec x$, as $y_i < x_i$. Let this element of $P$ rather be called $\vec x'$. If $\vec x' \leq_k (y_1, \ldots, y_{i-1},x_i-1,y_{i+1}, \ldots, y_k)$, then we have $x'_i \leq x_i-1 < x_i$. Since $\vec x$ and $\vec x'$ are incomparable, there must exist another $j \neq i$ such that $x'_j > x_j$ (as otherwise $\vec x'$ would dominate $\vec x$, which contradicts the fact that they are both Pareto points). But if $\vec y \leq_k \vec x$ (by the condition under which modified versions of $\vec y$ are added to $S'$), then by the fact that $x_j < x'_j$, we have $y_j \leq x_j < x'_j$. Now this fact contradicts the assumption that $\vec x' \leq_k (y_1, \ldots, y_{i-1},x_i-1,y_{i+1}, \ldots, y_k)$ (as $i \neq j$), so it did not hold in the first place. As the assumption that some point in $S'$ is dominated by some point in $P$ cannot hold, we proved part (b).
\end{proof}

\begin{lemma}
During every iteration of the main loop of Algorithm~\ref{algo:paretoEnumerator}, the set of points that dominate (or are) points in $S$ only shrinks, and it shrinks by at least one point per iteration.
\label{lem:SShrinkage}
\end{lemma}

\begin{proof}
There are two cases to be considered. If the algorithm finds that $f(\vec x) = \FALSE$ in line~\ref{line:checkSearchPoint}, then $\vec x$ is removed from $S$. Since $S \cup P$ is an antichain, $\vec x$ cannot dominate any other point in $S \cup P$, so if $\vec x$ is removed, then the set of points that are dominated by the ones in $S$ (or that are in $S$) shrinks by at least one point, namely $\vec x$.

Now let us consider the case that $f(\vec x) = \TRUE$ in line~\ref{line:checkSearchPoint}. After $\vec x$ is modified to be a Pareto point, all elements of $S$ are updated not to dominate $\vec x$. At the same time, no element is added to $S$ that is not dominated by any of the elements that were in $S$ previously. Thus, the set of search space points that dominate a point in $S$ (or are in $S$) strictly decreases.
\end{proof}

\begin{corollary}
At the end of the main loop of Algorithm~\ref{algo:paretoEnumerator}, $P$ contains the set of Pareto optima of $f$. As the algorithm always terminates (by Lemma~\ref{lem:SShrinkage}), this implies that Algorithm~\ref{algo:paretoEnumerator} computes all Pareto optima of a given feasibility function on the given search space.
\end{corollary}

\subsection{Efficiency}
Let us now analyze if Algorithm~\ref{algo:paretoEnumerator} is efficient. We consider two main questions in this context:
\begin{itemize}
\item Can the algorithm ever perform redundant oracle calls, i.e., does is ever
\begin{itemize}
\item query $f$ at some point $\vec x$ after it already got the result that $f(\vec y) = \TRUE$ for some $\vec y \leq_k \vec x$, or
\item query $f$ at some point $\vec x$ after it already got the result that $f(\vec y) = \FALSE$ for some $\vec x \leq_k \vec y$.
\end{itemize}
\item How many oracle calls does the algorithm need depending on $n$, $k$, and the number of Pareto points $p$? Is this number approximately optimal?
\end{itemize}
In the remainder of this subsection, we will discuss both of these questions.

\subsubsection{Redundant Oracle Calls}

\begin{lemma}
\label{lem:noRedundantCalls}
Algorithm~\ref{algo:paretoEnumerator} has the property that once a positive result has been obtained from $f$ for some point $\vec x$, then it never calls $f$ on some point $\vec y$ with $\vec x \leq_k \vec y$.
\end{lemma}
\begin{proof}
We prove this fact by induction over the iteration of Algorithm~\ref{algo:paretoEnumerator}'s main loop. We use the inductive hypothesis that the elements of $P$ dominate (or are equal to) all elements in the search space for which $f$ evaluated to $\TRUE$ in the preceding iterations of the main loop.

The induction basis is trivial. For the induction step, we need Lemma~\ref{lem:SandPInteraction}. As in line~\ref{line:checkSearchPoint} of the algorithm, $f$ is always evaluated at a point in $S$ and by Lemma~\ref{lem:SandPInteraction}, it cannot dominate an element in $P$, the claim holds for this check. If $f(\vec x)=\TRUE$ in that line, then function \textsc{SearchParetoPoint} is called, which is an iterative binary search procedure. 
In every individual dimension, this procedure does not perform redundant calls, as this is never the case in a binary search procedure, and we chose the concrete variant of the procedure such that $f$ is never evaluated at a maximal element in a singly-dimensional search range. 
The element of the search space that is the result of the iterative binary search procedure is smaller than or equal to all the search space points evaluated during the \textsc{SearchParetoPoint} function and the search space point given as a parameter to \textsc{SearchParetoPoint}. Thus, the inductive hypothesis also holds after the main loop iteration when the point returned by \textsc{SearchParetoPoint} has been added to $P$.
\end{proof}
Algorithm~\ref{algo:paretoEnumerator} does not satisfy a similar property for the case of negative answers by $f$. In particular, if the algorithm at some point determines that $f(\vec x) = \FALSE$ for some $\vec x$, then for some other $\vec y \leq \vec x$, it may still later test $f$ at point $\vec y$ even though it is implied by $f(\vec x) = \FALSE$ that the result of this test must be $\FALSE$ as well. The reason is the binary search in function \textsc{SearchParetoPoint}: since variable $\mathit{min}$ is always initialized to $0$, the first search points may be smaller than search points that were evaluated for other elements of $S$.

In order to avoid such redundant calls to $f$, the algorithm needs to be augmented by a \emph{wrapper} to function $f$ that caches search points $\vec x$ for which it is known that $f(\vec x)=\FALSE$. Whenever $f$ is then called on some $\vec y$ such that $\vec y \leq_k \vec x$, the wrapper returns $\FALSE$ without evaluating $f$. Due to Lemma~\ref{lem:noRedundantCalls}, only negative results need to be cached. To store the elements of the cache, a simple list can be used. To make caching more efficient, more advanced data structures to store Pareto points can also be used \cite{Mostaghim2005}. In this way, infeasible points that dominate other infeasible points can be removed from the cache, making the cache lookup process more efficient.

\subsection{Complexity Analysis}
\label{subsec:complexity}

Let us now analyze the efficiency of Algorithm~\ref{algo:paretoEnumerator}. We will only consider the number of oracle calls as efficiency measure and show that Algorithm~\ref{algo:paretoEnumerator} makes at most $p \cdot (k \cdot \lceil \log_2 n \rceil + 1) + \psi(p)$ many calls to the oracle function $f$.
We complement this result with the following lower bound lemmas: 
\begin{itemize}
\item Every algorithm that enumerates all elements of a Pareto front $P$ must also evaluate the objective function at all elements of the co-Pareto front induced by $P$.
\item When dividing the search space into $(\frac{n}{\log_2 n})^k$ large regions such that in each region, there is at most one Pareto front element, no two regions with Pareto points have comparable elements, and for all of the $p$ Pareto front elements, the region that it contains is known, $p \cdot k \cdot \log_2 \left(\frac{n}{\log_2 n} \right)$ calls to the oracle on non-co-Pareto front points are necessary in order to compute the exact positions of the Pareto front elements.
\end{itemize}
%When dividing the search space into $(\frac{n}{\log_2 n})^k$ large blocks such that every region contains at most one Pareto point and search space points from two regions with Pareto points are incomparable,
%, and no region with a Pareto point is adjacent to the search space boundaries, 
%no element of the co-Pareto front can be part of any of the regions with a Pareto point. 
As the oracle calls from the items above are independent, we get that a number of oracle calls of at least $p \cdot k \cdot \log_2 \left(\frac{n}{\log_2 n} \right) + \psi(p)$ is unavoidable to enumerate the Pareto front elements and to verify their completeness. 
As the general case (that the Pareto points are not known to be distributed in the stated way) is strictly more difficult, the lower bound also applies to the general case.
The lower bound is also very close to the $p \cdot (k \cdot \lceil \log_2 n \rceil + 1) + \psi(p)$ oracle call number upper bound of Algorithm~\ref{algo:paretoEnumerator}. This result however does not exclude that for dense Pareto sets, narrow search spaces, or low-dimensional search spaces, better algorithms exist, as a large high-dimensional search space is necessary in order to ensure that the Pareto points can be distributed in the way used for this lower bound.

% We start with a formal analysis of the the number of Algorithm~\ref{algo:paretoEnumerator}'s oracle calls.

\begin{lemma}
Given a search space of size $n^k$ and an objective function $f$ that admits $p$ Pareto optima over the search space, Algorithm~\ref{algo:paretoEnumerator} performs at most $p \cdot (k \cdot \lceil \log_2 n \rceil+1) + \psi(p)$ calls to $f$ before terminating.
\end{lemma}
\begin{proof}
First of all note that the main loop of the algorithm is executed at most $p + \psi(p)$ many times. This is because if in line~\ref{line:checkSearchPoint} of the algorithm, $f$ returns a value of $\TRUE$, then a new Pareto point is found during the same iteration of the loop, whereas if $f$ returns a value of $\FALSE$, the fact that all elements in the set $S$ form a set of largest points that are incomparable to the Pareto points found so far (see Lemma~\ref{lem:SandPInteraction}) implies that the point just evaluated is actually a co-Pareto point.

To obtain the stated efficiency, it remains to be shown that in line~\ref{line:searchParetoPoint} of the algorithm, no more than $k \cdot \lceil \log_2 n \rceil$ iterations take place. Function \textsc{SearchParetoPoint} iterates over all dimensions and for each of them performs a binary search over a singly-dimensional search space of size at most $n$, where it is already known that the maximal element in the (singly-dimensional) search space is a solution. If $\mathit{max}-\mathit{min}$ is a power of $2$, every iteration of the binary search algorithm causes one oracle call and reduces the local singly-dimensional search space by half, until the set of solutions is singleton. If $\mathit{max}-\mathit{min}$ is not a power of $2$, then the number of oracle calls is bounded by the next-higher power of two. Thus, for every dimension, function \textsc{SearchParetoPoint} performs at most $\lceil \log_2 n \rceil$ oracle calls.
\end{proof}

\begin{lemma}
Every algorithm that enumerates all elements of a Pareto front $P$ must also evaluate the objective function at all elements of the co-Pareto front induced by $P$.
\end{lemma}
\begin{proof}
Let $P'$ be the co-Pareto front. If an algorithm to enumerate $P$ terminates, then it must have checked that all points of $P$ have been found. In particular, for all elements $\vec p \in P'$, it must have evaluated $f$ at a point $\vec p'' \geq_k \vec p'$ at least once negatively (as only then, it is known that $\vec p'$ is not a feasible solution). Since all points $\vec p'' >_k \vec p'$ are dominated by some element in $P$ or are in $P$ (as $P'$ is the co-Pareto set), we have that $f(\vec p'') = \TRUE$, and hence, the only way to find out that $f(\vec p') = \FALSE$ for some $\vec p' \in P'$ is to evaluate $f$ on $\vec p'$.
\end{proof}

\begin{lemma}
When dividing the search space into $(\frac{n}{\log_2 n})^k$ large regions and for all of the $p$ Pareto front elements, the region that it contains is known and no two elements of different regions with Pareto points are comparable, $p \cdot k \cdot \log_2 \left(\frac{n}{\log_2 n} \right)$ calls to the oracle on non-co-Pareto front points are necessary in order to compute the exact positions of the Pareto front elements.
\end{lemma}

\begin{proof}
If all $(\frac{n}{\log_2 n})^k$ large blocks of the search space that contain Pareto points do not have comparable points, then the positions of the Pareto points are independent. Thus, all combinations of positions of the Pareto points in the blocks are possible. There are $\left((\frac{n}{\log_2 n})^k \right)^p$ different combinations of Pareto point positions. 
Since the feasibility function returns a $\FALSE$/$\TRUE$ result, this means that we need at least
$\log_2 \left(\left((\frac{n}{\log_2 n})^k \right)^p\right) = p \cdot k \cdot \log_2 \frac{n}{\log_2 n}$ evaluations of the feasibility function to find all Pareto points. None of these evaluations happen at co-Pareto points as by the fact that elements from two blocks with Pareto front elements are incomparable, we have that co-Pareto front elements do not dominate any point from any block with a Pareto front element. Thus, it is known already that $f$ maps these elements to $\FALSE$, and thus no information to locate the $p$ Pareto points can be obtained from evaluating the feasibility function on them.
\end{proof}

\section{Conclusion}
\label{sec:conclusion}

In this paper, we gave an algorithm to compute all Pareto front elements in a multiple-objective optimization problem in which all parameters are integer-valued (or finite and totally ordered). The algorithm uses a feasibility check for a candidate solution as an oracle and minimizes the number of calls to the oracle. The algorithm needs $p \cdot (k \cdot \lceil \log_2 n \rceil + 1) + \psi(p)$ many oracle calls, where $p$ is the number of Pareto optima, $k$ is the number of dimensions, $n$ is the number of values in each dimension, and $\psi(p)$ is the size of the co-Pareto-front of the problem. We have shown that every algorithm to solve this problem needs at least $p \cdot k \cdot \log_2 \frac{n}{\log_2 n} + \psi(p)$ many oracle calls (in the worst case), and hence our algorithm is close to optimal. 

The algorithm is useful to solve multi-criterial optimization problems for which the feasibility function is not known to be submodular and that cannot be encoded into a form for which specialized multicriterial optimizers are available. The algorithm is most suitable for sparse Pareto sets, and it is not difficult to construct better algorithms for very dense Pareto sets, in particular over a two-dimensional search space. Yet, the number of oracle calls of the algorithm is still relatively low, and it is well-suited to the case that nothing is known about the size of the Pareto set. Also, the algorithm is \newterm{any-time}, meaning that it computes the Pareto optimia successively, so that when it is aborted during runtime, at least some Pareto optima have been computed up to that point.

\section*{Acknowledgements}
This work was supported by the Institutional Strategy of the University of Bremen, funded by the German Excellence Initiative.

\bibliographystyle{eptcs}
\bibliography{bib}
\end{document}